\documentclass[conference]{IEEEtran}

\usepackage{graphicx}
\usepackage{bm}
\usepackage{enumerate}
\usepackage{float}
\usepackage{multicol}
\usepackage{color}
\usepackage{url}
\usepackage{cite}
\usepackage{flushend}

\usepackage{amsmath}
\usepackage{amsthm}
\usepackage{amssymb}
\DeclareMathOperator*{\argmax}{argmax}

\newtheorem{theorem}{Theorem}
\usepackage{algorithmic}
\usepackage{algorithm}
\usepackage{acronym}
\usepackage{array}

\ifCLASSOPTIONcompsoc
  \usepackage[caption=false,font=normalsize,labelfont=sf,textfont=sf]{subfig}
\else
  \usepackage[caption=false,font=footnotesize]{subfig}
\fi

\acrodef{BS}{base station}
\acrodef{RA}{receive antenna}
\acrodef{PA}{predictor antenna}
\acrodef{IID}{independent and identically distributed}
\acrodef{CDF}{cumulative distribution function}
\acrodef{PDF}{probability density function}
\acrodef{cu}{channel use}
\acrodef{ACK}{acknowledgement}
\acrodef{NACK}{negative acknowledgement}
\acrodef{SNR}{signal-to-noise ratio}
\acrodef{HARQ}{hybrid automatic repeat request}
\acrodef{CSIT}{channel state information at the transmitter side}
\acrodef{INR}{incremental redundancy}
\acrodef{npcu}{nats-per-channel-use}
\acrodef{MIMO}{multiple input multiple output }
\acrodef{FDD}{frequency division duplex}
\acrodef{E2E}{end-to-end}
\acrodef{IAB}{integrated access and backhaul}
\acrodef{3GPP}{3rd Generation Partnership Project}

\begin{document}
\captionsetup{belowskip=0pt,aboveskip=0pt}

\title{Predictor Antennas for Moving Relays: \\Finite Block-length Analysis}
\author{\IEEEauthorblockN{ Hao Guo}
\IEEEauthorblockA{\textit{Department of Electrical Engineering} \\
\textit{Chalmers University of Technology}\\
Gothenburg, Sweden \\
hao.guo@chalmers.se}
\and
\IEEEauthorblockN{ Behrooz Makki}
\IEEEauthorblockA{\textit{Department of Ericsson Research} \\
\textit{Ericsson Research}\\
Gothenburg, Sweden\\
behrooz.makki@ericsson.com}
\and
\IEEEauthorblockN{ Tommy Svensson}
\IEEEauthorblockA{\textit{Department of Electrical Engineering} \\
\textit{Chalmers University of Technology}\\
Gothenburg, Sweden\\
tommy.svensson@chalmers.se}
}


\maketitle

\begin{abstract}
In future wireless networks, we anticipate that a large number of devices will connect to mobile networks through moving relays installed on vehicles, in particular in public transport vehicles. To provide high-speed moving relays with accurate channel state information different methods have been proposed, among which predictor antenna (PA) is one of the promising ones. Here, the PA system refers to a setup where two sets of antennas are deployed on top of a vehicle, and the front antenna(s) can be used to predict the channel state information for the antenna(s) behind. In this paper, we study the delay-limited performance of PA systems using adaptive rate allocations. We use the fundamental results on the achievable rate of finite block-length codes to study the system throughput and error probability in the presence of short packets. Particularly, we derive closed-form expressions for the error probability, the average transmit rate as well as the optimal rate allocation, and study the effect of different parameters on the performance of PA systems. The results indicate that rate adaptation under finite block-length codewords can improve the performance of the PA system with spatial mismatch. 
\end{abstract}

\begin{IEEEkeywords}
6G, channel estimation, finite block-length analysis, integrated access and backhaul (IAB), moving relay, predictor antenna, relay, throughput, vehicle-to-everything (V2X), wireless backhaul.
\end{IEEEkeywords}

\section{Introduction}
Moving relay, as a complement to fixed relays, has been shown to be a promising way to serve mobile users \cite{yutao2013moving}. One of the key challenges for moving relays is, however, how to obtain accurate \ac{CSIT} since the relay terminal is often under mobility. To overcome this issue, the concept of \ac{PA} has been proposed in \cite{Sternad2012WCNCWusing} where two groups of antennas are placed on top of a vehicle: the 
\ac{PA} positioned in the front is used to predict the \ac{CSIT} for the antennas behind it, referred to as \ac{RA}, who will later experience the similar environment as the \ac{PA} (see Fig. \ref{fig_mismatch}). Different testbed based studies have proved the validation of such concept embedded with technologies such as \ac{MIMO}  \cite{DT2015ITSMmaking,phan2018WSAadaptive} and Kalman smoothing \cite{Apelfrojd2018PIMRCkalman}. Particularly, these field trials  indicate that the \ac{PA} system can obtain \ac{CSIT}  with high accuracy also at high speed. Such \ac{CSIT} is key for enabling closed loop transmission techniques such as link adaptation, multi-user scheduling, advanced \ac{MIMO} schemes and interference coordination also at high speed. Thus, the PA concept is a promising  candidate for mobile relays in future generation of wireless networks.  This is important specially because, with the recent \ac{3GPP} standardization progress in \ac{IAB} networks \cite{madapatha2020integrated}, it is not unlikely that mobile \ac{IAB}, as an efficient method of relaying, will be considered in the future \ac{3GPP} releases.

If the \ac{RA} reaches the same point as the place where the \ac{PA} was estimating the channel in the previous time slots, the \ac{CSIT} will be perfect. However, due to, e.g., varying vehicle speed or the processing delay at the \ac{BS}, the \ac{RA} may not reach the ideal point. As a result, the acquired \ac{CSIT} may be inaccurate, which will affect the system performance correspondingly. 

Such spatial mismatch can be compensated by channel interpolation with dense time samples at the cost of computational complexity  and reduced end-to-end throughput \cite{Apelfrojd2018PIMRCkalman,BJ2017PIMRCpredictor}. Alternatively, one can utilize the spatial correlations  based on imperfect \ac{CSIT}. This idea has been exploited in \cite{Guo2019WCLrate,guo2020power,guo2020semilinear,Guo2020rate}, where the spatial correlations have been utilized to optimize the transmit power \cite{guo2020power} and rate \cite{Guo2019WCLrate,guo2020semilinear,Guo2020rate}. Moreover, in \cite{guo2020power,Guo2020rate} the \ac{PA} is not only used for channel estimation, but also for data transmission, and adaptive power/rate allocation combined with \ac{HARQ} shows notable performance improvement. However,  the analysis in \cite{Guo2019WCLrate,guo2020power,guo2020semilinear,Guo2020rate} assumes infinite codeword length, which may not always be a valid assumption in vehicular networks,  where the codewords are designed to be short such that the Shannon’s capacity formula does not give an accurate approximation of the achievable rate \cite{polyanskiy2010channel}. Specially, as reported in, e.g., \cite{bilstrup2008evaluation},  in vehicle communications the codeword length is in the order of several hundred symbols. Hence, it is interesting to investigate the system performance for the cases with short packets.

\begin{figure}
    \centering
    \includegraphics[width = 1\columnwidth]{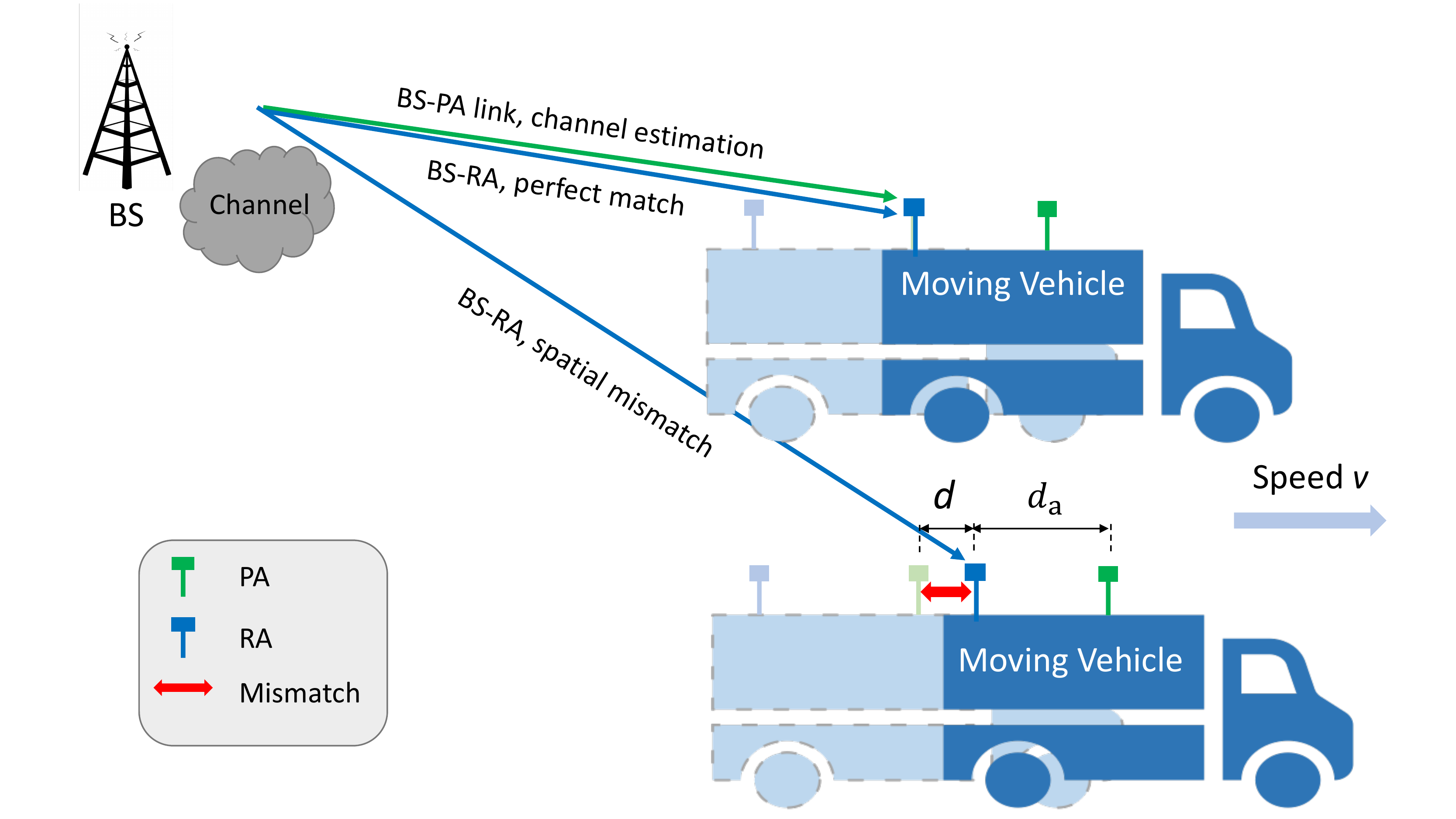}
    \caption{PA system with spatial mismatch.}
    \label{fig_mismatch}
\end{figure}

In this paper, we analysis the data transmission efficiency of the \ac{PA} systems with spatial mismatch utilizing adaptive rate allocation. We apply the fundamental results of \cite{polyanskiy2010channel,Makki2014WCLfinite} on the error probability of finite block-length codes to investigate the delay-limited throughput and error probability of the \ac{PA} systems. Particularly, considering codewords of finite length, we determine closed-form expressions for the error probability, the system throughput as well as the appropriate rate allocation maximizing the throughout. Finally, we investigate the effect of different parameters such as the antennas distance and the vehicle speed on the performance of \ac{PA} systems.  In this way, considering short packets leads to completely different analytical results compared to \cite{Sternad2012WCNCWusing,DT2015ITSMmaking,phan2018WSAadaptive,Apelfrojd2018PIMRCkalman,BJ2017PIMRCpredictor,Guo2019WCLrate,guo2020power,guo2020semilinear,Guo2020rate}, and makes it possible to study the effect of the codeword length on the system performance.

The analytical and simulation results show that using finite-length codewords and rate adaptation there is potential for improving the performance of the \ac{PA} system with spatial mismatch. Also, while the system throughput and error probability are sensitive to the length of short packets, their sensitivity decreases as the codeword length increases.

\section{System Model}
Consider a PA setup  with one \ac{PA} in the front of the vehicle  and one \ac{RA} aligned behind the \ac{PA}. The \ac{PA}   is used for channel estimation for the \ac{BS}-\ac{RA} link. We assume a \ac{FDD}-based setup where, at $t_1$ the \ac{BS} sends pilots to the \ac{PA}. Then, the \ac{PA} estimates the \ac{BS}-\ac{PA} channel $\hat{h}$ and sends it back to the \ac{BS}. Here, we assume the feedback link from the PA to the BS to be perfect. This is motivated by the fact that, compared to the direct link, considerably lower rate is required in the feedback channel \cite{makki2013TCfeedback}. However, one can follow the same scheme as in \cite{guo2020semilinear} to model the effect of \ac{CSIT} quantization on the system performance (see Section III). Finally, at $t_3$ the \ac{BS} sends the data to the \ac{RA} using the channel information from the \ac{PA}, and the signal received by the \ac{RA} can be expressed as
\begin{align}
    y = hx + n.
\end{align}
Here, $h$ is the \ac{BS}-\ac{RA} channel and $n$ represents additive Gaussian noise with zero mean and unit variance. Ideally, $h = \hat{h}$ if  the \ac{RA} reaches the same  point as the \ac{PA} when it was sending pilots. However, as demonstrated in Fig. \ref{fig_mismatch}, if the \ac{RA} does not reach the  position where the \ac{PA} estimated the channel, due to, e.g., the \ac{BS} processing time is not equal to the time we need until the RA reaches the same point as the PA sending pilots, spatial mismatch may happen such that $h \neq \hat{h}$. Nevertheless, $h$ and $\hat{h}$ are spatially correlated. Using the Jake's model with uniform scattering \cite{Shin2003TITcapacity}, the correlation between $h$ and $\hat{h}$ can be modeled as \cite{Guo2019WCLrate}
\begin{align}\label{eq_hhat}
    \bigl[ \begin{smallmatrix}
  \hat{h}\\h
\end{smallmatrix} \bigr] = \bm{\Phi}^{1/2}  \bigl[ \begin{smallmatrix}
  p\\q
\end{smallmatrix} \bigr].
\end{align}
Here, $p, q\sim \mathcal{CN}(0,1)$ which are independent of $\hat{h}$ and $h$. Also, $\bm{\Phi}^{1/2}$ is a $2\times2$ correlation matrix with the ($i,j$)-th entity given by
\begin{align}\label{eq_phi}
    \Phi_{i,j} = J_0\left((i-j)\cdot2\pi d/ \lambda\right),
\end{align}
where $J_0(x) = \sum\limits_{m=0}^{\infty} \frac{(-1)^{m}}{m!\Gamma(m+1)}\left(\frac{x}{2}\right)^{2m}$ is the zeroth-order Bessel function of the first kind with $\Gamma(z) = \int_0^{\infty} x^{z-1}e^{-x} \mathrm{d}x$ representing the Gamma function, and $\lambda=c/f_\text{c}$ represents the wavelength  with $c$ being the speed of light and $f_\text{c}$ being the carrier frequency. Moreover, as illustrated in Fig. \ref{fig_mismatch}, $d$ in (\ref{eq_phi}) represents the mismatch distance 
\begin{align}\label{eq_d}
    d = |d_\text{a} - d_\text{m} | = |d_\text{a} - v\delta|,
\end{align}
which is  the difference between the antenna separation $d_\text{a}$ between the \ac{PA} and the \ac{RA} and the moving distance $d_\text{m}=v\delta$ where $v$ represents the vehicle speed. Also, $\delta = t_2-t_1$ denotes the processing delay at the \ac{BS}. 

Then, plugging (\ref{eq_phi}) and (\ref{eq_d}) into (\ref{eq_hhat}), $h$ can be expressed as a function of $\hat{h}$ following 
\begin{align}\label{eq_h}
    h = \sqrt{1-\sigma^2} \hat{h} + \sigma q.
\end{align}
Here, $\sigma$ is the (1,2)-th entity of $\bm{\Phi}^{1/2}$ with normalization and it is a function of the mismatch distance $d$ and carrier wavelength $\lambda$. In this way,  for a given $\hat{h}$ and $\sigma \neq 0$, $|h|$ follows a Rician distribution, i.e., the  \ac{PDF} of $|h|$ is given by 
\begin{align}
    f_{|h|\big|\hat{g}}(x) = \frac{2x}{\sigma^2}e^{-\frac{x^2+(1-\sigma^2)\hat{g}}{\sigma^2}}I_0\left(\frac{2x\sqrt{(1-\sigma^2)\hat{g}}}{\sigma^2}\right),
\end{align}
where $\hat{g} = |\hat{h}|^2$. Let us define the channel gain between BS-RA as $ g = |{h}|^2$. Then, the \ac{PDF} of $g|\hat{g}$ is given by
\begin{align}\label{eq_pdf}
    f_{g|\hat{g}}(x) = \frac{1}{\sigma^2}e^{-\frac{x+(1-\sigma^2)\hat{g}}{\sigma^2}}I_0\left(\frac{2\sqrt{x(1-\sigma^2)\hat{g}}}{\sigma^2}\right),
\end{align}
which is non-central Chi-squared distributed with the \ac{CDF} 
\begin{align}\label{eq_cdf}
    F_{g|\hat{g}}(x) = 1 - \mathcal{Q}_1\left( \sqrt{\frac{2(1-\sigma^2)\hat{g}}{\sigma^2}}, \sqrt{\frac{2x}{\sigma^2}}  \right).
\end{align}
Here, $ \mathcal{Q}_1(s,\rho) = \int_{\rho}^{\infty} xe^{-\frac{x^2+s^2}{2}}I_0(sx)\,\text{d}x$ is  the first-order Marcum $Q$-function \cite{Bocus2013CLapproximation} where $I_n(x) = (\frac{x}{2})^n \sum\limits_{i=0}^{\infty}\frac{(\frac{x}{2})^{2i} }{i!\Gamma(n+i+1)}$ is the $n$-th order modified Bessel function of the first kind.

\section{Analytical Results}\label{Sec. III}
Considering the data transmission with finite block-length codewords,  the error probability in a time slot is given by \cite{polyanskiy2010channel,Makki2014WCLfinite}
\begin{align}\label{eq_QL}
    \epsilon = Q\left(\frac{\sqrt{L}\left(\log\left(1+gP\right)-R\right)}{\sqrt{1-\frac{1}{\left(1+gP\right)^2}}}\right).
\end{align}
Here, $L$ in channel use is the codeword length, $P$ is the transmit power, and $Q(x) = \frac{1}{\sqrt{2\pi}}\int_x^{\infty}e^{-\frac{t^2}{2}}\text{d}t$ represents the Gaussian Q-function. Also, $R=\frac{K}{L}$ is the codeword length in \ac{npcu} with $K$ being the number of nats per codeword.

With the PA system, if we focus on one slot of the BS-RA channel with perfect knowledge  of the BS-PA channel $\hat{g}$ at the \ac{BS}, the instantaneous error probability is given by
\begin{align}\label{eq_Qi}
\epsilon =\underset{g|\hat{g}}{\mathbb{E}}\left[ Q\left(\frac{\sqrt{L}\left(\log\left(1+gP\right)-R(\hat{g})\right)}{\sqrt{1-\frac{1}{\left(1+gP\right)^2}}}\right)\right],
\end{align}
i.e., the transmission rate can be dynamically optimized based on the quality of the BS-PA channel. Also, the instantaneous throughput (in \ac{npcu}), defined as the average number of information nats successfully received by the RA for the given BS-PA channel realization $\hat g$, is given by
\begin{align}\label{eq_etaghat}
    \eta_{|\hat {g}}&=  R(\hat{g})\left(1-\epsilon\right).
\end{align}
Moreover, with a fixed $R$, the error probability is given by
\begin{align}\label{eq_Q_fixedR}
\epsilon =\underset{g}{\mathbb{E}}\left[ Q\left(\frac{\sqrt{L}\left(\log\left(1+gP\right)-R\right)}{\sqrt{1-\frac{1}{\left(1+gP\right)^2}}}\right)\right].
\end{align}

Let us define $\alpha = \frac{e^{R(\hat{g})}-1}{P}$, and 
\begin{align}\label{eq_mu}
   \mu &= -\frac{\partial \left(Q\left(\frac{\sqrt{L}\left(\log\left(1+xP\right)-R(\hat{g})\right)}{\sqrt{1-\frac{1}{\left(1+xP\right)^2}}}\right)\right)}{\partial x}\Bigg|_{x=\alpha} \nonumber\\
   &= \sqrt{\frac{LP^2}{2\pi\left(e^{2R}-1\right)}},
\end{align}
which is the derivative of the Gaussian Q-function in (\ref{eq_Qi}) at $x = \alpha$. Then, the Gaussian Q-function in  (\ref{eq_Qi}) can be linearly approximated as \cite[Eq. 14]{Makki2014WCLfinite}
\begin{align}\label{eq_linear}
&Q\left(\frac{\sqrt{L}\left(\log\left(1+gP\right)-R(\hat{g})\right)}{\sqrt{1-\frac{1}{\left(1+gP\right)^2}}}\right)\simeq \nonumber\\&
   \left\{ \begin{array}{rcl}
1 & \mbox{for} & x<\alpha-\frac{1}{2\mu}, \\ 
\frac{1}{2}-\mu(x-\alpha) & \mbox{for} & \alpha-\frac{1}{2\mu}\leq x \leq \alpha+\frac{1}{2\mu}, \\
0 & \mbox{for} & x>\alpha+\frac{1}{2\mu}.
\end{array}\right.
\end{align}
Our goal is to derive a closed-form expression for (\ref{eq_etaghat}) and (\ref{eq_Q_fixedR}) which allows for optimal rate allocation. For this reason, Theorems \ref{theorem1} and \ref{theorem2} use (\ref{eq_linear}) to approximate the conditional error probability (\ref{eq_Qi}) as follows.

\begin{theorem}\label{theorem1}
The error probability (\ref{eq_Qi}) for a given $\hat{g}$  can be  approximated by (\ref{eq_theorem1}).
\end{theorem}
\begin{proof} Plugging (\ref{eq_pdf}) into (\ref{eq_Qi}), we obtain
\begin{align}\label{eq_theorem1}
    \epsilon &= \int_0^\infty f_{g|\hat{g}}(x) Q\left(\frac{\sqrt{L}\left(\log\left(1+gP\right)-R(\hat{g})\right)}{\sqrt{1-\frac{1}{\left(1+gP\right)^2}}}\right) \text{d}x \nonumber\\
    &\overset{(a)}{\simeq} \int_0^{\alpha-\frac{1}{2\mu}} f_{g|\hat{g}}(x) \text{d}x + \int_{\alpha-\frac{1}{2\mu}}^{\alpha+\frac{1}{2\mu}} f_{g|\hat{g}}(x)\left(\frac{1}{2}-\mu(x-\alpha)\right) \text{d}x \nonumber\\
    &=  F_{g|\hat{g}}\left(\alpha-\frac{1}{2\mu}\right) +\nonumber\\
    &~~\left(\frac{1}{2}+\mu\alpha\right)\left(F_{g|\hat{g}}\left(\alpha+\frac{1}{2\mu}\right)-F_{g|\hat{g}}\left(\alpha-\frac{1}{2\mu}\right)\right)-\nonumber\\
    &~~\mu \int_{\alpha-\frac{1}{2\mu}}^{\alpha+\frac{1}{2\mu}} \frac{x}{\sigma^2}e^{-\frac{x+(1-\sigma^2)\hat{g}}{\sigma^2}}I_0\left(\frac{2\sqrt{x(1-\sigma^2)\hat{g}}}{\sigma^2}\right) \text{d}x \nonumber\\
    &\overset{(b)}{\simeq}  F_{g|\hat{g}}\left(\alpha-\frac{1}{2\mu}\right) +\nonumber\\
    &~~\left(\frac{1}{2}+\mu\alpha\right)\left(F_{g|\hat{g}}\left(\alpha+\frac{1}{2\mu}\right)-F_{g|\hat{g}}\left(\alpha-\frac{1}{2\mu}\right)\right)-\nonumber\\
    &~~\mu \frac{e^{-\frac{\left(1-\sigma^2\right)\hat{g}}{\sigma^2}}}{\sigma^2}\left(\sum_{i=0}^{N}\frac{\left(\left(1-\sigma^2\right)\hat{g}\right)^{i}}{\sigma^{4i}\left(i!\right)^2}\int_{\alpha-\frac{1}{2\mu}}^{\alpha+\frac{1}{2\mu}}x^{i+1}e^{-\frac{x}{\sigma^2}}\text{d}x\right)\nonumber\\
    &\overset{(c)}{\simeq}  F_{g|\hat{g}}\left(\alpha-\frac{1}{2\mu}\right) +\nonumber\\
    &~~\left(\frac{1}{2}+\mu\alpha\right)\left(F_{g|\hat{g}}\left(\alpha+\frac{1}{2\mu}\right)-F_{g|\hat{g}}\left(\alpha-\frac{1}{2\mu}\right)\right)-\nonumber\\
    &~~\mu \frac{e^{-\frac{\left(1-\sigma^2\right)\hat{g}}{\sigma^2}}}{\sigma^2}\sum_{i=0}^{N}\frac{\left(\left(1-\sigma^2\right)\hat{g}\right)^{i}}{\sigma^{2i-4}\left(i!\right)^2}\nonumber\\
    &~~\left(\Gamma\left(i+2, \frac{\alpha-\frac{1}{2\mu}}{\sigma^2}\right)-\Gamma\left(i+2, \frac{\alpha+\frac{1}{2\mu}}{\sigma^2}\right)\right).
\end{align}
Here, $(a)$ is obtained by (\ref{eq_linear}), and $(b)$ comes from the approximation $I_0(x) \simeq \sum_{i=0}^{N} \frac{\left(\frac{x^2}{4}\right)^{i}}{\left(i!\right)^2}, \forall N$. Finally, $(c)$ uses the definition of the incomplete Gamma function $\Gamma(s,x) = \int_x^{\infty} t^{s-1} e^{-t} \text{d}t$.
\end{proof}

\begin{theorem}\label{theorem2}
The error probability (\ref{eq_Qi}) for a given $\hat{g}$ can be approximated by (\ref{eq_theorem2}).
\end{theorem}
\begin{proof}
Using (\ref{eq_linear}), we have
\begin{align}\label{eq_theorem2}
    \epsilon &= \int_0^\infty f_{g|\hat{g}}(x) Q\left(\frac{\sqrt{L}\left(\log\left(1+gP\right)-R(\hat{g})\right)}{\sqrt{1-\frac{1}{\left(1+gP\right)^2}}}\right) \text{d}x \nonumber\\
    &\simeq F_{g|\hat{g}}\left(\alpha-\frac{1}{2\mu}\right) +\nonumber\\
    &~~\left(\frac{1}{2}+\mu\alpha\right)\left(F_{g|\hat{g}}\left(\alpha+\frac{1}{2\mu}\right)-F_{g|\hat{g}}\left(\alpha-\frac{1}{2\mu}\right)\right)-\nonumber\\
    &~~\mu \int_{\alpha-\frac{1}{2\mu}}^{\alpha+\frac{1}{2\mu}} xf_{g|\hat{g}}(x) \text{d}x \nonumber\\
    &\overset{(d)}{=}F_{g|\hat{g}}\left(\alpha-\frac{1}{2\mu}\right) +\nonumber\\
    &~~\left(\frac{1}{2}+\mu\alpha\right)\left(F_{g|\hat{g}}\left(\alpha+\frac{1}{2\mu}\right)-F_{g|\hat{g}}\left(\alpha-\frac{1}{2\mu}\right)\right)-\nonumber\\
    &~~\mu\left(F_{g|\hat{g}}(x)x|_{\alpha-\frac{1}{2\mu}}^{\alpha+\frac{1}{2\mu}}-\int_{\alpha-\frac{1}{2\mu}}^{\alpha+\frac{1}{2\mu}}F_{g|\hat{g}}(x)\text{d}x\right)\nonumber\\
    &\overset{(e)}{\simeq} F_{g|\hat{g}}\left(\alpha-\frac{1}{2\mu}\right) +\nonumber\\
    &~~\left(\frac{1}{2}+\mu\alpha\right)\left(F_{g|\hat{g}}\left(\alpha+\frac{1}{2\mu}\right)-F_{g|\hat{g}}\left(\alpha-\frac{1}{2\mu}\right)\right)-\nonumber\\
    &~~\mu\Bigg(\left(\alpha+\frac{1}{2\mu}\right)F_{g|\hat{g}}\left(\alpha+\frac{1}{2\mu}\right)-\nonumber\\&~~\left(\alpha-\frac{1}{2\mu}\right)F_{g|\hat{g}}\left(\alpha-\frac{1}{2\mu}\right)\nonumber\\&~~-\frac{1}{\mu}F_{g|\hat{g}}\left(\alpha-\frac{1}{2\mu}\right)+\frac{1}{\mu}F_{g|\hat{g}}\left(\alpha-\frac{1}{2\mu}\right)-\nonumber\\&~~\frac{1}{\mu}F_{g|\hat{g}}\left(\frac{1}{2}2\alpha\right)\Bigg)=F_{g|\hat{g}}(\alpha).
\end{align}
Here, following the same approach as in \cite[Eq. (34)]{makki2016TCwireless}, ($d$) is obtained by partial integration, and in ($e$) we use the first-order Riemann integral approximation $\int_{t_0}^{t_1}f(x)\text{d}x \simeq (t_1-t_0)f\left(\frac{t_0+t_1}{2}\right)$ and some manipulations.
\end{proof}

By Theorem \ref{theorem2}, the outage-limited throughput for a given $\hat{g}$ is simplified to
\begin{align}\label{eq_etagivenhatg}
    \eta|_{\hat{g}} &=  R(\hat{g})\left(1-F_{G|\hat{H}}\left(\frac{e^{R(\hat{g})}-1}{P}\right)\right)\nonumber\\
    & = R(\hat{g})\mathcal{Q}_1\left(\sqrt{\frac{2(1-\sigma^2)\hat{g}}{\sigma^2}},\sqrt{\frac{2(e^{R(\hat{g})}-1)}{P\sigma^2}}\right),
\end{align}
which is the same as the case with infinite block length. Hence, we can use the same method as in \cite[Eq. (14)]{Guo2019WCLrate} 
with the approximation of the first-order Marcum Q-function \cite[Eq. (2), (7)]{Bocus2013CLapproximation}
\begin{align}
    \mathcal{Q}_1 (s, \rho) &\simeq e^{\left(-e^{\mathcal{I}(s)}\rho^{\mathcal{J}(s)}\right)}, \nonumber\\
    \mathcal{I}(s)& = -0.840+0.327s-0.740s^2+0.083s^3-0.004s^4,\nonumber\\
    \mathcal{J}(s)& = 2.174-0.592s+0.593s^2-0.092s^3+0.005s^4,
\end{align}
the definition of the Lambert W function $xe^x = y \Leftrightarrow x = \mathcal{W}(y)$ \cite{corless1996lambertw},
and 
\begin{align}
    \omega(\sigma, \hat{g}) &= e^{\left(\mathcal{I}\left(\sqrt{\frac{2(1-\sigma^2)\hat{g}}{\sigma^2}}\right)\right)}\left(\frac{2}{P\sigma^2}\right)^{\frac{\mathcal{J}\left(\sqrt{\frac{2(1-\sigma^2)\hat{g}}{\sigma^2}}\right)}{2}}, \\
    \nu(\sigma, \hat{g}) &= \frac{\mathcal{J}\left(\sqrt{\frac{2(1-\sigma^2)\hat{g}}{\sigma^2}}\right)}{2}.
\end{align}
to approximate (\ref{eq_etagivenhatg}) as 
\begin{align}\label{eq_appR}
    \eta|_{\hat{g}} \simeq R(\hat{g})e^{\left(-\omega(\sigma, \hat{g})\left(e^{R(\hat{g})}-1\right)^{\nu(\sigma, \hat{g})}\right)}.
\end{align}
Then, setting the derivative of (\ref{eq_appR}) with respect to $R(\hat{g})$ equal to zero leads to the optimal instantaneous rate allocation given by 
\begin{align}\label{eq_appRF}
    R_{\text{opt}|\hat{g}} \simeq \frac{1}{\nu(\sigma, \hat{g})}\mathcal{W}\left(\frac{1}{\omega(\sigma, \hat{g})}\right).
\end{align}

Finally, the average throughput for all possible values of $\hat{g}$ is given by 
\begin{align}\label{eq_averageeta}
    \eta = \mathbb{E}_{\hat{g}}\left[R_{\text{opt}|\hat{g}}\mathcal{Q}_1\left(\sqrt{\frac{2(1-\sigma^2)\hat{g}}{\sigma^2}},\sqrt{\frac{2(e^{R_{\text{opt}|\hat{g}}}-1)}{P\sigma^2}}\right)\right].
\end{align}

In this way, one can use (\ref{eq_appRF}) to optimize the rate allocation based on imperfect \ac{CSIT}, caused by the spatial mismatch problem, and (\ref{eq_averageeta}) gives the average system performance. Finally,  with adaptive rate allocation, the average error probability  can be calculated by averaging (\ref{eq_Qi}) for all possible values of $\hat{g}$.

As benchmarks, in the following, we consider two extreme cases with no \ac{CSIT} and no spatial mismatch.

\subsection{No-\ac{CSIT} Scenario}

With no instantaneous CSIT and rate adaptation based on average Rayleigh-fading at the transmitter side,  the average throughput is given by
\begin{align}\label{eq_eta_no_csit_pre}
    \eta^{\text{no-CSIT}} =& R\left(1-\underset{g}{\mathbb{E}}\left[Q\left(\frac{\sqrt{L}\left(\log\left(1+gP\right)-R\right)}{\sqrt{1-\frac{1}{\left(1+gP\right)^2}}}\right)\right]\right)\nonumber\\
    =& R\left(1-\int_0^\infty e^{-x}Q\left(\frac{\sqrt{L}\left(\log\left(1+xP\right)-R\right)}{\sqrt{1-\frac{1}{\left(1+xP\right)^2}}}\right) \text{d}x \right)\nonumber\\
    &\overset{(h)}{\simeq} \mu R \left(e^{-\alpha+\frac{1}{2\mu}}-e^{-\alpha-\frac{1}{2\mu}}\right)\nonumber\\
    & \overset{(i)}{\simeq}\sqrt{\frac{LP^2}{2\pi}}e^{-R-\frac{e^{R}-1}{P}}Re^{\sqrt{\frac{e^{2R}\pi}{2LP^2}}}.
\end{align}
Here, $(h)$ uses the semi-linear approximation of the Gaussian Q-function (\ref{eq_linear}), and $(i)$ is obtained by $\mu\simeq\sqrt{\frac{LP^2}{2\pi}}e^{-R}$, and ignoring the small term $e^{-\alpha-\frac{1}{2\mu}}$ with some manipulations.

Then, setting the derivative of (\ref{eq_eta_no_csit_pre}) with respect to $R$ equal to zero, we obtain the optimal rate adaptation without CSIT as
\begin{align}\label{eq_R_no_csit}
    R^{\text{no-CSIT}}_{\text{opt}} &\overset{(j)}{\simeq} \underset{R}{\arg}\Bigg\{\left(P-2\sqrt{\frac{LP^2}{2\pi}}\right)Re^{R}-\nonumber\\&~~2\sqrt{\frac{LP^2}{2\pi}}PR+2\sqrt{\frac{LP^2}{2\pi}}P=0\Bigg\}\nonumber\\
    &=\underset{R}{\arg}\left\{\left(\sqrt{\frac{\pi}{2LP^2}}-\frac{1}{P}\right)Re^{R}-R+1 = 0\right\}\nonumber\\
    &\overset{(k)}{\simeq} \mathcal{W}\left(\frac{P}{1-\sqrt{\frac{L}{2\pi}}}\right).
\end{align}
Here, $(j)$ is obtained by $e^{R}-1\simeq e^{R}$ for moderate/large \acp{SNR}, and in $(k)$ we omit the second term in the left hand of the equation. Also, $\mathcal{W}(\cdot)$ is the Lambert W function.

\subsection{No Spatial Mismatch Scenario}
In a genie-aided case with no spatial mismatch, we update the transmit rate in each time slot and target on a fixed instantaneous error probability as
\begin{align}
  Q\left(\frac{\sqrt{L}\left(\log\left(1+gP\right)-R\right)}{\sqrt{1-\frac{1}{\left(1+gP\right)^2}}}\right) = \hat{\epsilon}.
\end{align}
In this way, denoting the inverse Gaussian Q-function by $Q^{-1}(\cdot)$, the instantaneous transmit rate with fixed error probability is given by
\begin{align}
    R_{\text{ins}} = \log(1+gP) - \frac{Q^{-1}(\hat{\epsilon})}{\sqrt{L}}\sqrt{1-\frac{1}{(1+gP)^2}}.
\end{align}
Then, the expected throughput is calculated as
\begin{align}\label{eq_eta_fullcsit_pre}
    \eta_{\text{no spatial mismatch}} = \mathbb{E}_{g} \left[R_{\text{ins}}\right](1-\hat{\epsilon}).
\end{align}

Under Rayleigh-fading conditions, (\ref{eq_eta_fullcsit_pre}) can be calculated as \cite[Section IV.B]{makki2015TCfinite}
\begin{align}\label{eq_eta_fullcsit_pre2}
    \eta_{\text{no spatial mismatch}} = \left(R^{L\rightarrow\infty}-\frac{Q^{-1}(\hat{\epsilon})}{\sqrt{L}}\zeta\right)(1-\hat{\epsilon}),
\end{align}
with 
\begin{align}
   R^{L\rightarrow\infty} = \int_0^{\infty} e^{-x}\log(1+Px)\text{d}x = e^{\frac{1}{P}}\operatorname{E_1}\left(\frac{1}{P}\right), 
\end{align}
which is equal to the expected achievable rate under the assumption $L\rightarrow\infty$, and  $\operatorname{E_1}(x) = \int_x^{\infty} \frac{e^{-t}}{t} \mathrm{d}t$ denotes the Exponential Integral function. Also,
\begin{align}
    \zeta = \int_0^{\infty} e^{-x} \sqrt{{1-\frac{1}{(1+Px)^2}}}\text{d}x
\end{align}
represents the finite block-length penalty term and can be calculated numerically.

In this way, the optimal rate allocation with perfect \ac{CSIT} is derived from
\begin{align}
     \hat{\epsilon} &\overset{(l)}{=} \underset{x}{\argmax}\left\{\log\left(1-Q(x)\right)+\log\left(1-\frac{\zeta x}{\sqrt{L}R^{L\rightarrow\infty}}\right)\right\}\nonumber\\&\overset{(m)}{\simeq}\underset{x}{\argmax}\left\{Q(x)+\frac{\zeta x}{\sqrt{L}R^{L\rightarrow\infty}}\right\}\nonumber\\
     &\overset{(n)}{=} \underset{x}{\arg} \left\{\frac{-e^{-\frac{x^2}{2}}}{\sqrt{2\pi}}+\frac{\zeta x}{\sqrt{L}}R^{L\rightarrow\infty} = 0\right\}\nonumber\\
     &\overset{(o)}{=} Q\left(\sqrt{-2\log\left(\frac{\sqrt{2\pi}\zeta}{\sqrt{L}R^{L\rightarrow\infty}}\right)}\right).
\end{align}
Here, $(l)$ uses variable transform $x = Q^{-1}(\hat{\epsilon})$ and takes the logarithm of (\ref{eq_eta_fullcsit_pre2}). Then, $(m)$ uses approximation $\log(1-x)\simeq x$ for small values of $x$ and $(n)$ is obtained by setting the derivative of the previous step with respect to $x$ equal to zero. Finally, $(o)$ is obtained by variable transform $x = Q^{-1}(\hat{\epsilon})$ and some manipulations.

In this way, the optimal throughput under full-CSIT is given by
\begin{align}\label{eq_eta_fullcsit}
    \eta_{\text{no spatial mismatch}} = \left(R^{L\rightarrow\infty}-\frac{\sqrt{-2\log\left(\frac{\sqrt{2\pi}\zeta}{\sqrt{L}R^{L\rightarrow\infty}}\right)}}{\sqrt{L}}\zeta\right)\nonumber\\
    \left(1-Q\left(\sqrt{-2\log\left(\frac{\sqrt{2\pi}\zeta}{\sqrt{L}R^{L\rightarrow\infty}}\right)}\right)\right).
\end{align}
Finally, note that the cases with no spatial mismatch is an ultimate upper bound of the system performance, as in practise there will always be some spatial mismatch due to, e.g., speed variation.


\section{{Simulation Results}}
\begin{figure}
    \centering
    \includegraphics[width = 1.0\columnwidth]{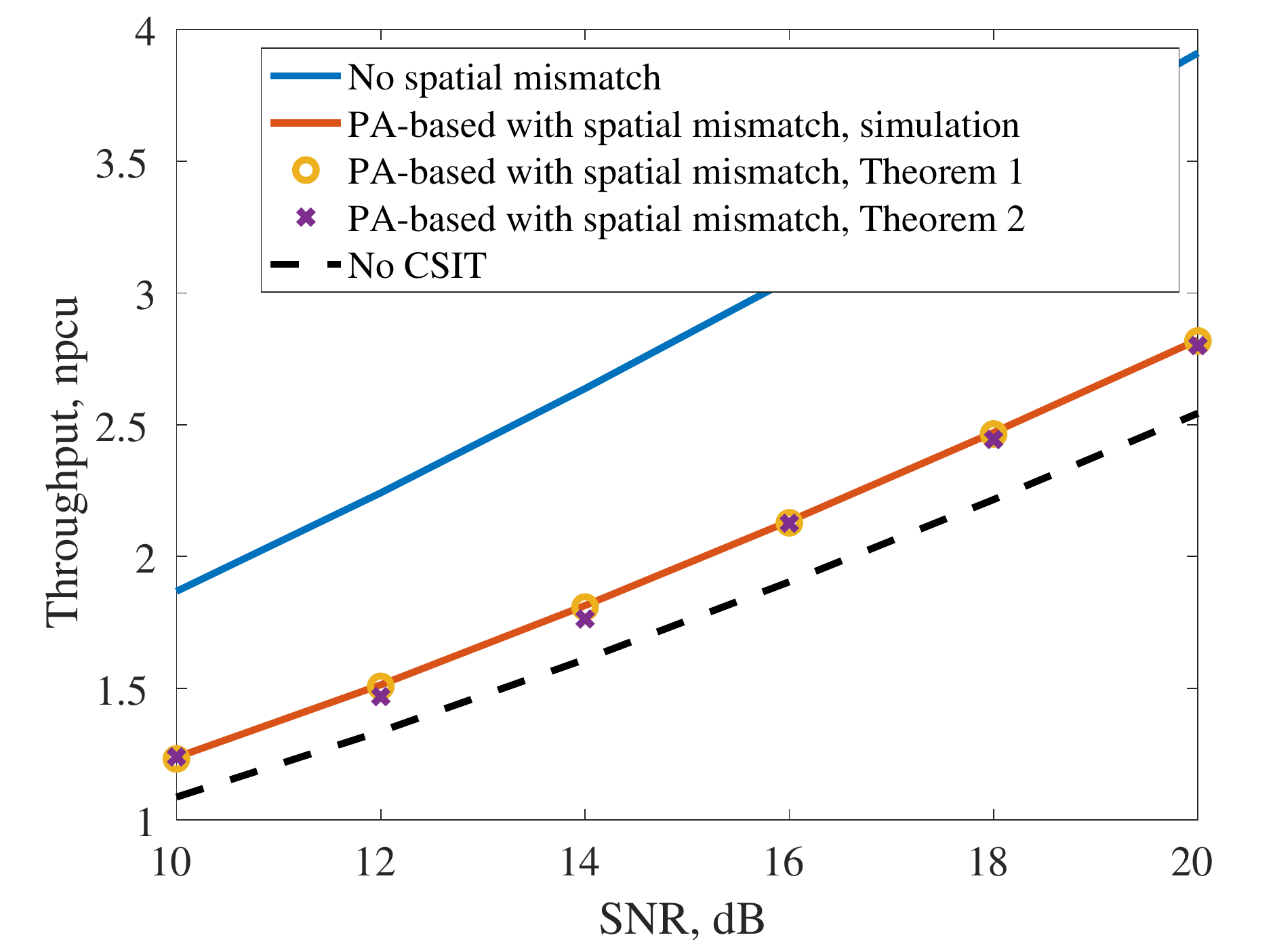}
    \caption{Optimal throughput (\ref{eq_averageeta}) as a function of \ac{SNR} for different setups with $\sigma = 0.5$, $L = 300$.}
    \label{fig2}
\end{figure}

\begin{figure}
    \centering
    \includegraphics[width = 1.0\columnwidth]{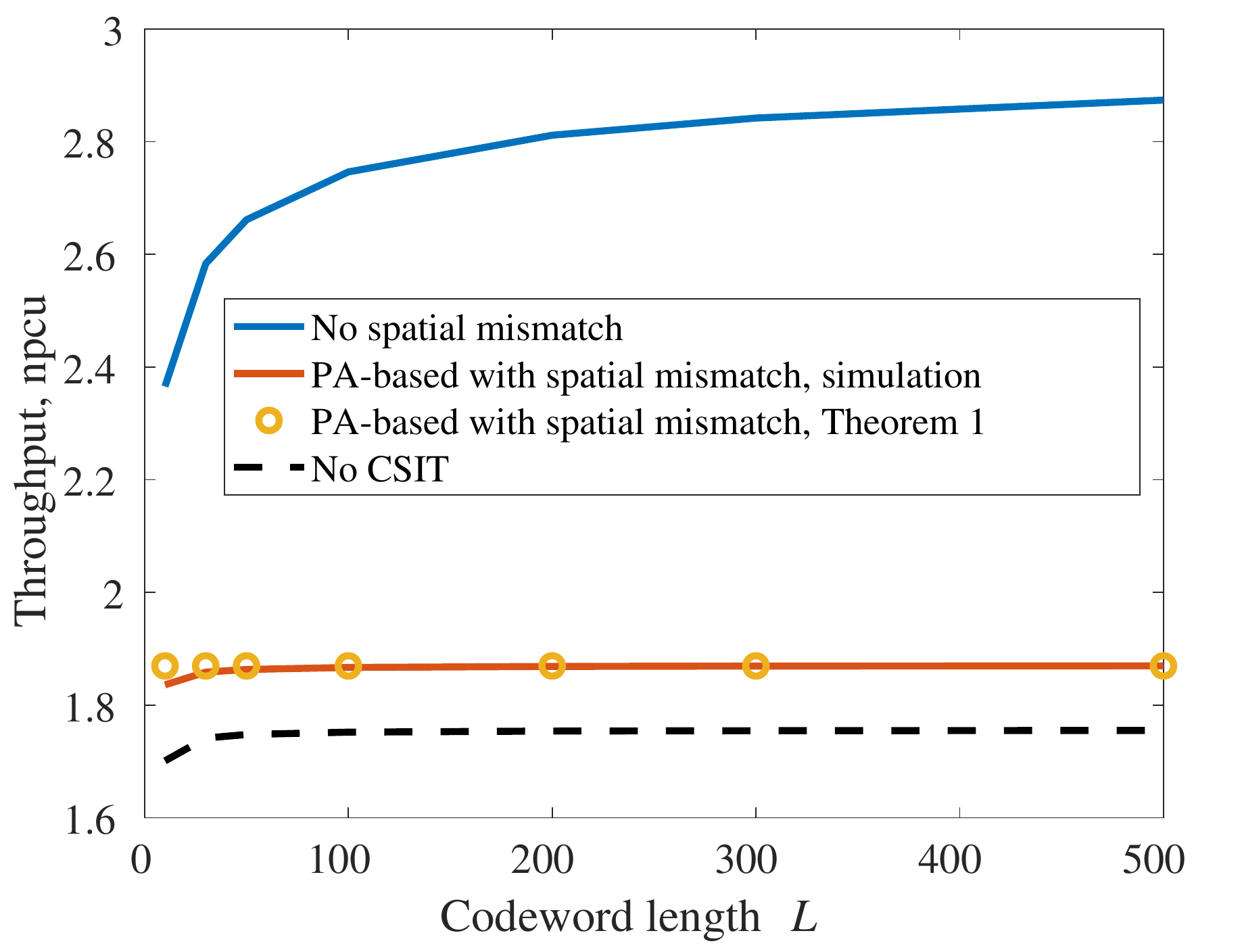}
    \caption{Optimal throughput (\ref{eq_averageeta}) as a function of codeword length $L$ with $\sigma = 0.6$, \ac{SNR} $ = 15$ dB.}
    \label{fig3}
\end{figure}

\begin{figure}
    \centering
    \includegraphics[width = 1.0\columnwidth]{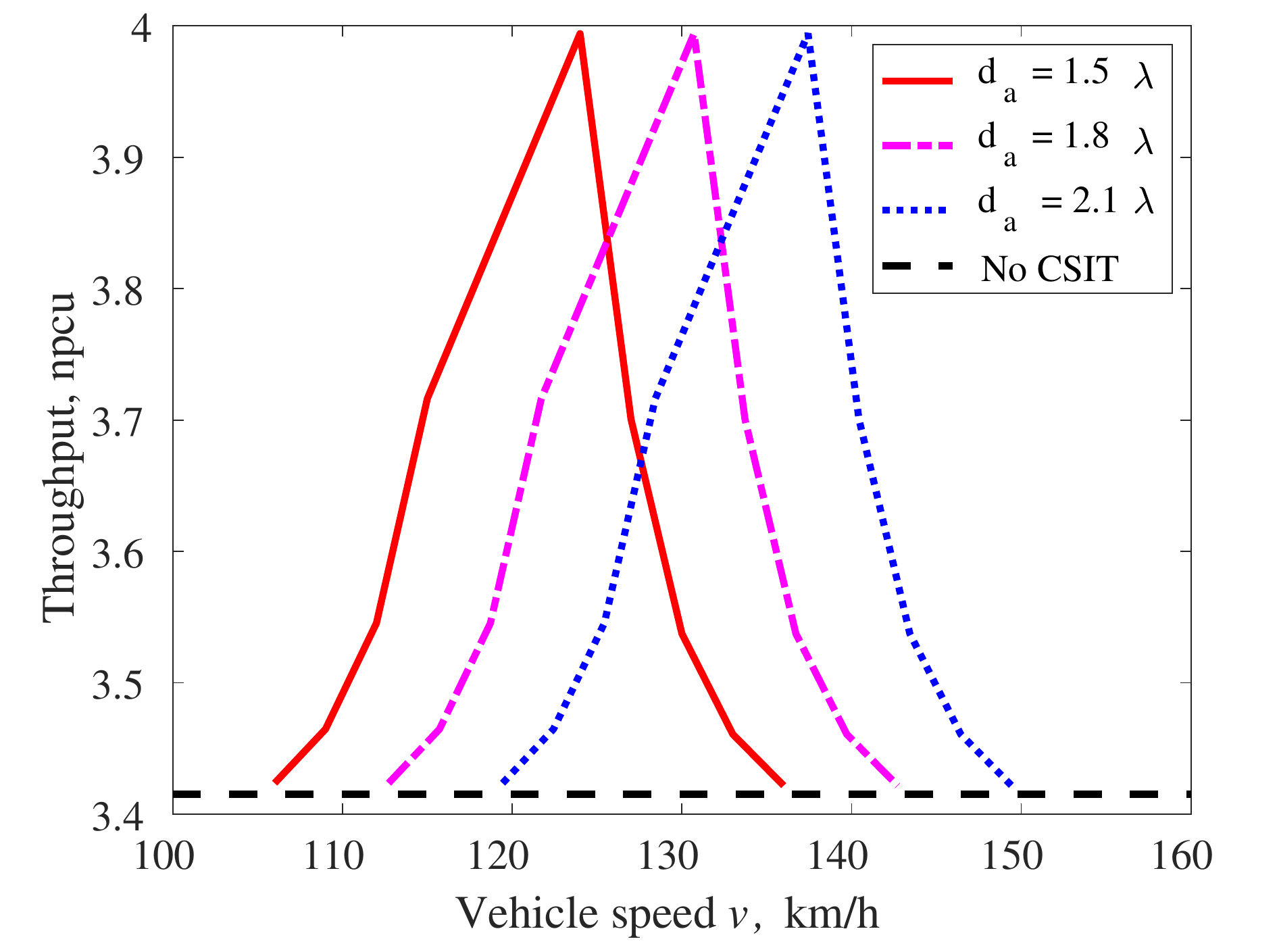}
    \caption{Optimal throughput (\ref{eq_averageeta}) as a function of the vehicle speed $v$ and antenna separation $d_\text{a}$. Here, we set $ \delta$ = 5 ms, carrier frequency $f_\text{c}$ = 2.68 GHz. Also, \ac{SNR} = 25 dB and $L = 300$ channel use. }
    \label{fig5}
\end{figure}
\begin{figure}
    \centering
    \includegraphics[width = 1.0\columnwidth]{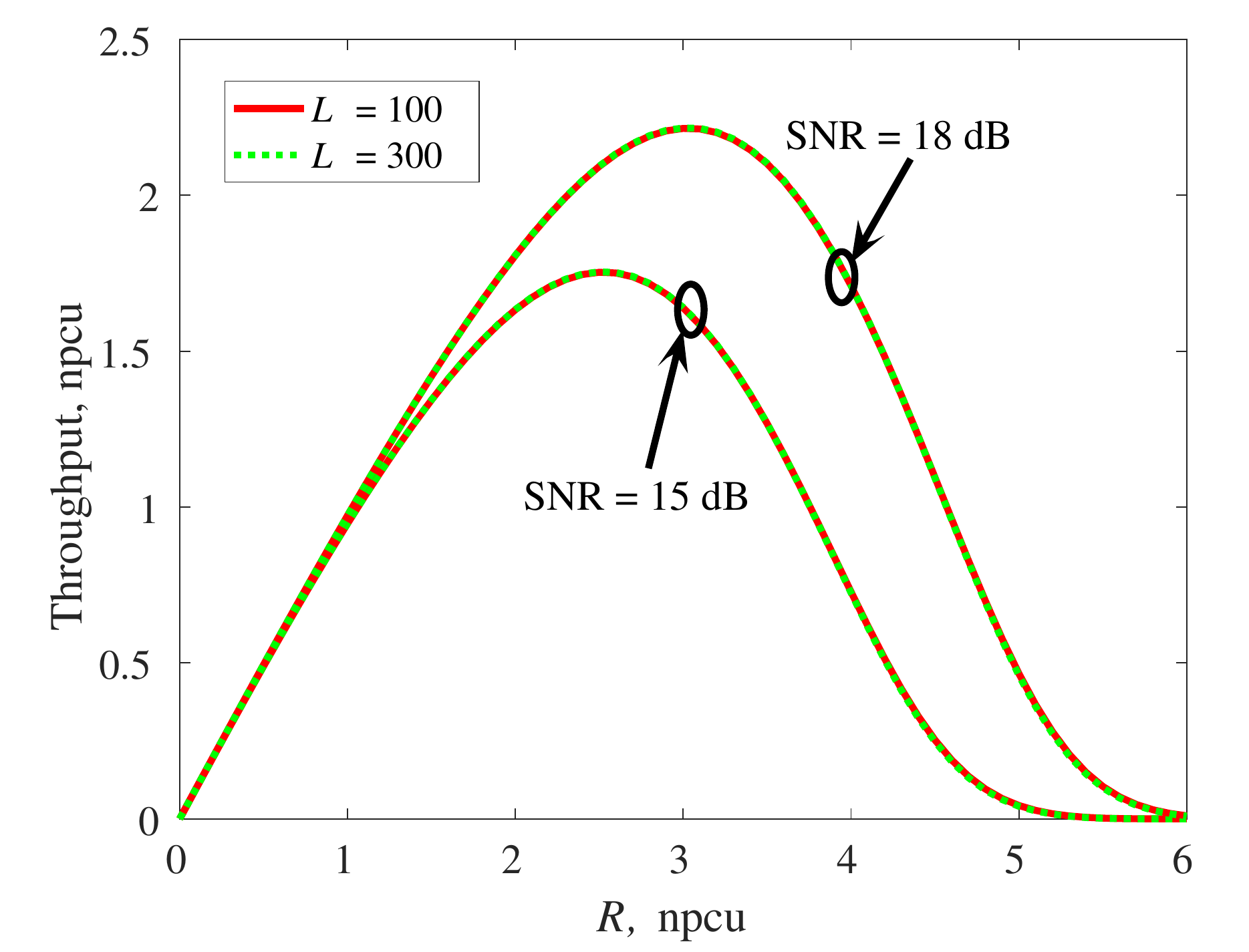}
    \caption{Throughput (\ref{eq_averageeta}) with fixed transmit rate. Here, we set \ac{SNR} = 15, 18 dB, $\sigma =  0.5$, and $L = $100, 300 channel use.}
    \label{fig6}
\end{figure}
In the simulations, we study the performance of rate adaptation in the \ac{PA} systems with spatial mismatch considering finite block-length codewords. The optimal throughput (\ref{eq_averageeta}) is from both Theorems \ref{theorem1}-\ref{theorem2}. For Theorem \ref{theorem1}, we find optimal throughput with (\ref{eq_theorem1}) for given $\hat{g}$ and average over all possible values of  $\hat{g}$, while for Theorem \ref{theorem2} we use the optimal transmit rate from (\ref{eq_appRF}) for a given $\hat{g}$ and then average the performance over different values of $\hat{g}$. Also, the cases with no \ac{CSIT} and no spatial mismatch are given by  (\ref{eq_eta_no_csit_pre}) and (\ref{eq_eta_fullcsit}), respectively. 

Setting $\sigma$ = 0.5, $L = 300$ channel use, Fig. \ref{fig2} studies the optimal throughput (\ref{eq_averageeta}) as a function of the \ac{SNR} for the cases with no spatial mismatch, \ac{PA}-based partial \ac{CSIT} due to spatial mismatch, and no \ac{CSIT}. Here, as the noise variance is set to 1, we define the \ac{SNR}, in dB, as $10\log_{10}P.$ Then, in Fig. \ref{fig3}, we study the average throughput as a function of the codeword length $L$ with $\sigma = 0.6$ and \ac{SNR} = 15 dB. Then, in Fig. \ref{fig5} we study the sensitivity of the optimal throughput with speed variation for $ \delta$ = 5 ms, carrier frequency $f_\text{c}  = 2.68$ GHz, \ac{SNR} = 25 dB and $L = 300$ channel use. In Fig. \ref{fig6},  we show the throughput as a function of the transmit rate $R$ (without rate adaptation) for different \acp{SNR} (SNR = 15, 18 dB) and codeword lengths ($L = 100, 300$ channel use)  and $\sigma = 0.5$.

In order to study the system performance in terms of error probability, in Fig. \ref{fig8} we present the error probability as a function of the transmit rate $R$ (without rate adaptation) for different \acp{SNR} (SNR = 15, 18 dB) and codeword lengths ($L = 100, 300$ channel use)  with $\sigma = 0.3$.  To evaluate the effect of the codeword length on the error probability (\ref{eq_QL}), in Fig. \ref{fig4} we present the error probability as a function of the codeword length $L$ for the cases with $\sigma = 0.4$ and \ac{SNR}$ = 15$ dB.  Finally, in Fig. \ref{fig7} the average error probability is presented as a function of the vehicle speed and the \ac{SNR}. Here, we set $L = 300$ channel use, and \ac{SNR} = 15, 25 dB.

\begin{figure}
    \centering
    \includegraphics[width = 1.0\columnwidth]{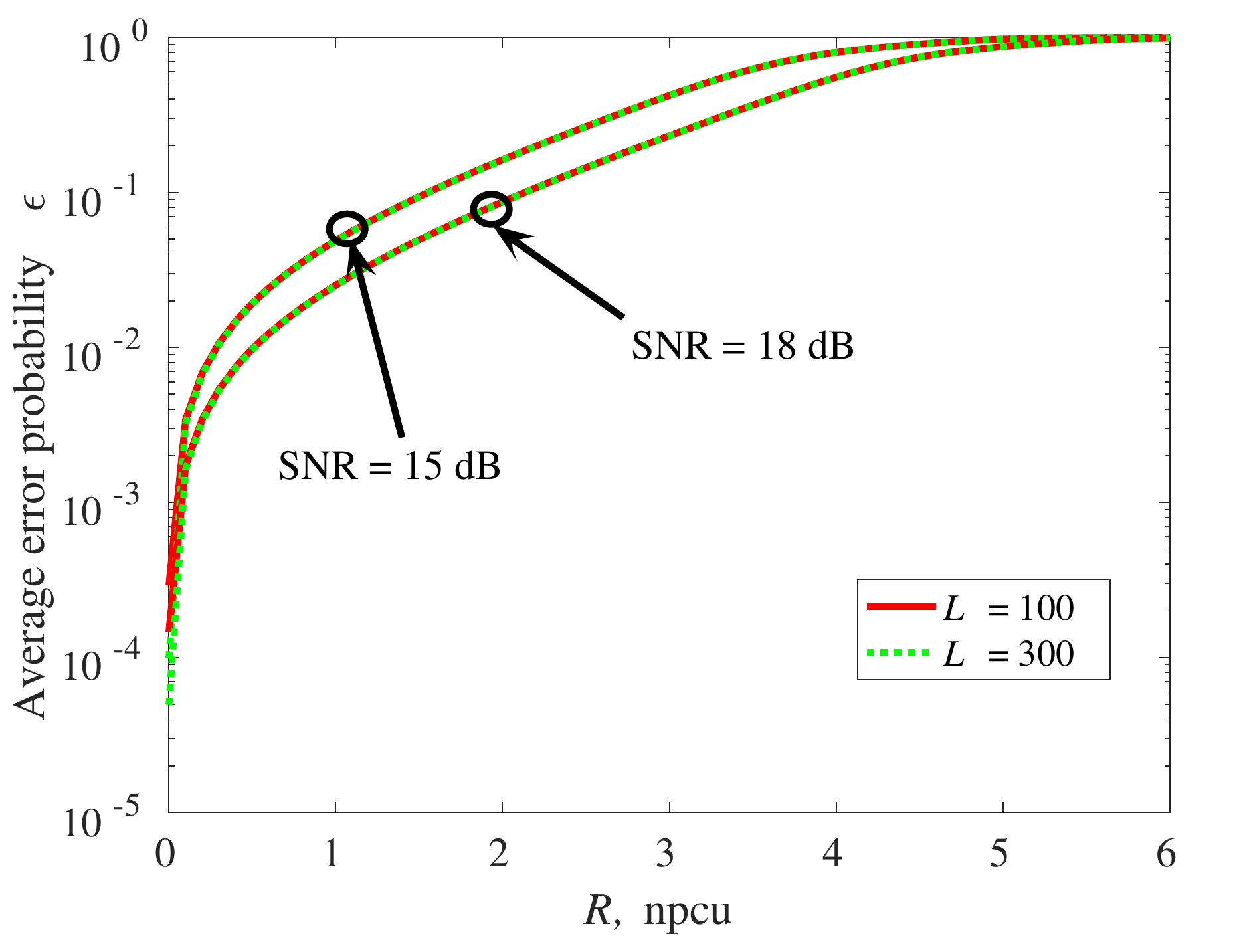}
    \caption{Average error probability with fixed transmit rate. Here, we set \ac{SNR} = 15, 18 dB, $\sigma = 0.3$, and $L = $100, 300 channel use.}
    \label{fig8}
\end{figure}
\begin{figure}
    \centering
    \includegraphics[width = 1.0\columnwidth]{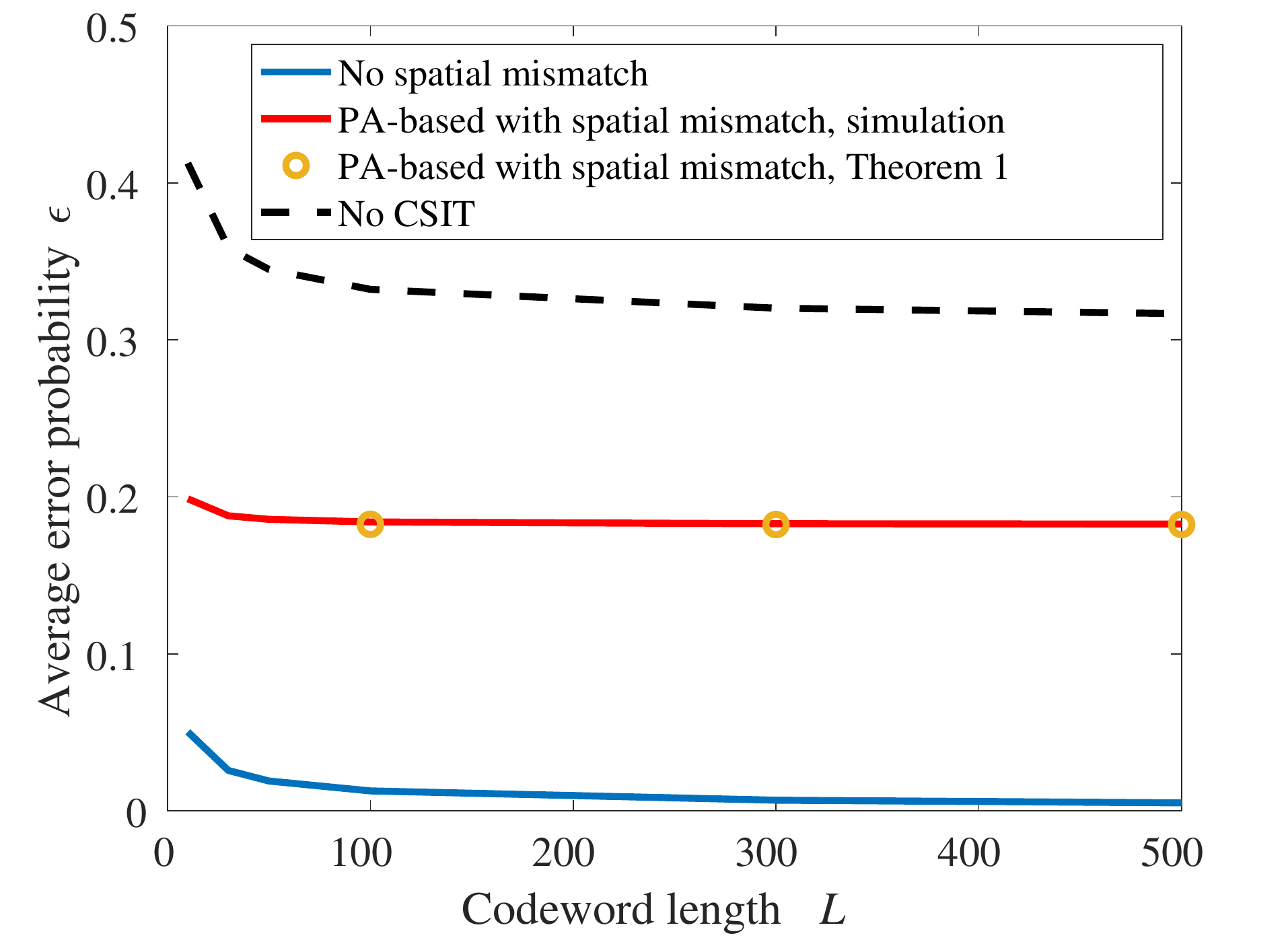}
    \caption{Average error probability with adaptive rate allocation. $\sigma = 0.4$, \ac{SNR}$ = 15$ dB.}
    \label{fig4}
\end{figure}
\begin{figure}
    \centering
    \includegraphics[width = 1.0\columnwidth]{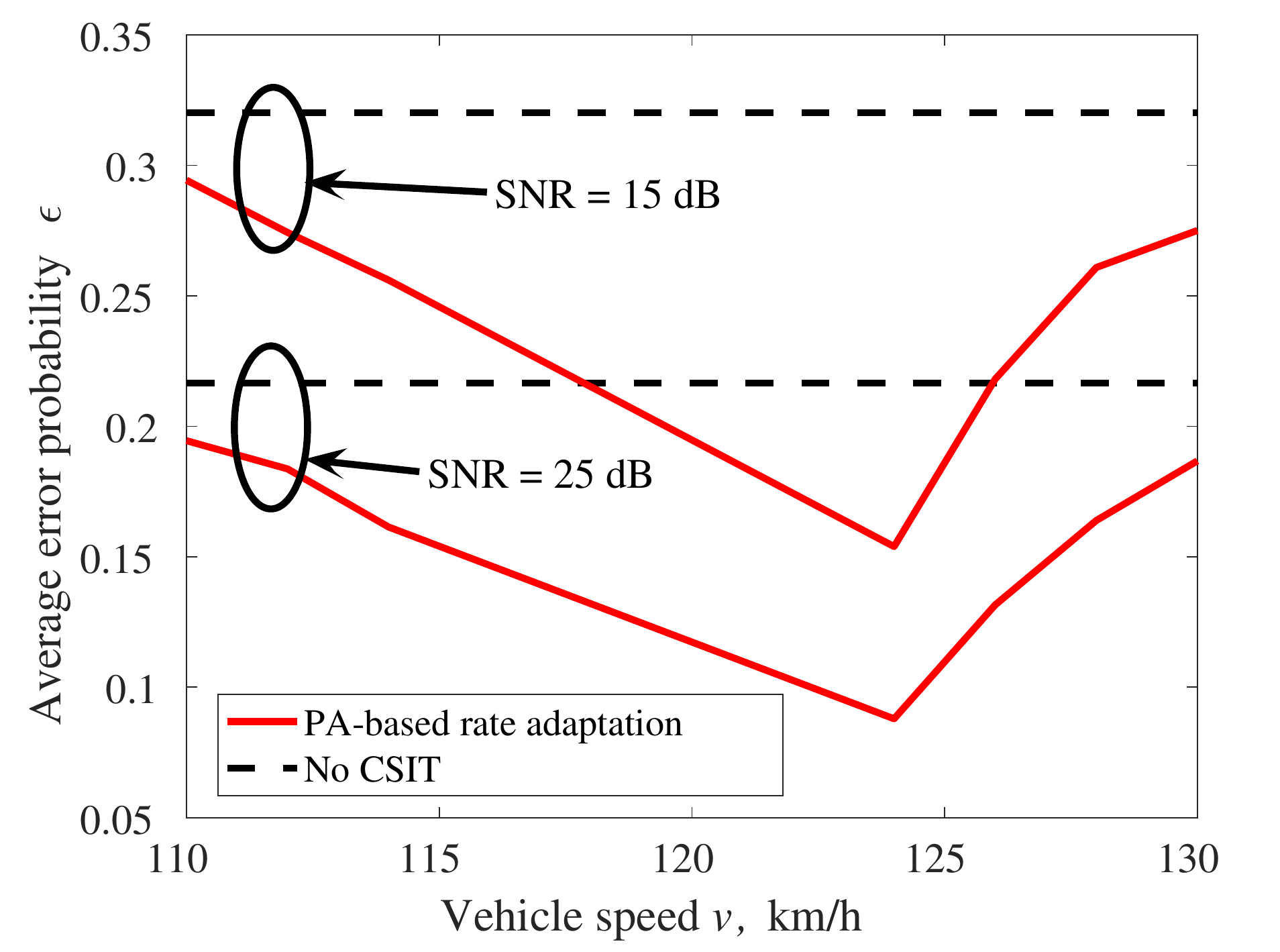}
    \caption{Average error probability with adaptive rate allocation. Here, we set \ac{SNR} = 15, 25 dB,  $ \delta$ = 5 ms, $f_\text{c}$ = 2.68 GHz, and $L = $300 channel use.}
    \label{fig7}
\end{figure}

According to the figures, the following conclusions can be drawn:
\begin{itemize}
    \item For a broad range of parameter settings, the approximation results of Theorems \ref{theorem1} and \ref{theorem2} give an accurate approximation of the error probability and the throughput (Figs. \ref{fig2}, \ref{fig3}, \ref{fig4}). Thus, they can be well utilized to evaluate the delay-limited performance of \ac{PA} setups.
    \item Rate adaptation can remarkably improve the performance of the \ac{PA} system with spatial mismatch in terms of throughput and error probability, specially in the cases with moderate/high \acp{SNR} (Fig. \ref{fig2}), long codeword lengths (Fig. \ref{fig3}, \ref{fig4}), as well as small $\sigma$'s (Figs. \ref{fig5}, \ref{fig7}).
    \item As shown in Fig. \ref{fig6}, for different values of codeword lengths/\acp{SNR}, there is an optimal value for the transmission rate maximizing the throughput. This is intuitively because with a low transmission rate the message is (almost) always correctly decoded but there are few information nats received by the receiver. With a high rate, on the other hand, there is a high probability of unsuccessful decoding, as shown in Fig. \ref{fig8}. Thus, the throughput tends to zero at low and high transmission rates, and there is a finite value of the transmission rate optimizing the throughput.
    \item The effect of the codeword length $L$ on the throughput and the error probability is relatively small according to Figs. \ref{fig3}, \ref{fig6} and \ref{fig4}, especially when the codeword length is in the order of several hundred symbols, which are the regions of interest in vehicle communications \cite{bilstrup2008evaluation}. In this way, unless for the cases with very short codewords, the Shannon's capacity formula can be well applied to study the system performance in the cases with moderate/long codeword length.
    \item As seen in Figs. \ref{fig5} and \ref{fig7}, the performance of the \ac{PA} system, in terms of throughput and error probability, is sensitive to the vehicle speed $v$, which leads to different levels of spatial mismatch for given $d_\text{a}$, $f_\text{c}$ and $\delta$, according to (\ref{eq_d}). As shown in Fig. \ref{fig5}, the throughput increases with higher correlations of the \ac{BS}-\ac{PA} and \ac{BS}-\ac{RA} channels, i.e., smaller $\sigma$, compared to the cases with no \ac{CSIT}. Also, the optimal speed in terms of maximum throughput for given $\delta$ changes with different antenna separations. The similar behaviour can be observed in Fig. \ref{fig7} as well. Also, it is interesting that for a range of speeds,  \ac{PA} with 10 dB less \ac{SNR} leads to better error probability compared to open-loop  setup.  Finally, while the error probability with rate adaptation improves compared to the case with no CSIT, it is sensitive to speed as well as \ac{SNR}. In this way, the \ac{PA} setup not only boosts the throughput, but it also increases the reliability.
\end{itemize}

\section{Conclusion}
We studied the performance of the \ac{PA} systems in the presence of spatial mismatch and  finite block-length codewords. The simulation and analytical results show that, while adaptive rate allocation improves the system performance in terms of throughput and error probability considerably, the effectiveness is affected by the \ac{SNR} and the vehicle speed. but less sensitive to the codeword length especially at the region of modern vehicle communications. However, the sensitivity of the system performance to the codeword length is negligible in the cases with codewords of moderate/long length which is the range of interest in vehicle communication systems. 

\section*{Acknowledgement}
This work was supported by VINNOVA (Swedish Government Agency for Innovation Systems) within the VINN Excellence Center ChaseOn.

\bibliographystyle{IEEEtran}

\bibliography{main.bib}

\begin{thebibliography}{10}
\providecommand{\url}[1]{#1}
\csname url@samestyle\endcsname
\providecommand{\newblock}{\relax}
\providecommand{\bibinfo}[2]{#2}
\providecommand{\BIBentrySTDinterwordspacing}{\spaceskip=0pt\relax}
\providecommand{\BIBentryALTinterwordstretchfactor}{4}
\providecommand{\BIBentryALTinterwordspacing}{\spaceskip=\fontdimen2\font plus
\BIBentryALTinterwordstretchfactor\fontdimen3\font minus
  \fontdimen4\font\relax}
\providecommand{\BIBforeignlanguage}[2]{{%
\expandafter\ifx\csname l@#1\endcsname\relax
\typeout{** WARNING: IEEEtran.bst: No hyphenation pattern has been}%
\typeout{** loaded for the language `#1'. Using the pattern for}%
\typeout{** the default language instead.}%
\else
\language=\csname l@#1\endcsname
\fi
#2}}
\providecommand{\BIBdecl}{\relax}
\BIBdecl

\bibitem{yutao2013moving}
S.~Yutao, J.~Vihriala, A.~Papadogiannis, M.~Sternad, W.~Yang, and T.~Svensson,
  ``Moving cells: a promising solution to boost performance for vehicular
  users,'' \emph{IEEE Commun. Mag.}, vol.~51, no.~6, pp. 62--68, Jun. 2013.

\bibitem{Sternad2012WCNCWusing}
M.~Sternad, M.~Grieger, R.~Apelfr\"ojd, T.~Svensson, D.~Aronsson, and A.~B.
  Martinez, ``Using predictor antennas for long-range prediction of fast fading
  for moving relays,'' in \emph{Proc. IEEE WCNCW}, Paris, France, Apr. 2012,
  pp. 253--257.

\bibitem{DT2015ITSMmaking}
D.-T. Phan-Huy, M.~Sternad, and T.~Svensson, ``{Making 5G adaptive antennas
  work for very fast moving vehicles},'' \emph{IEEE Intell. Transp. Syst.
  Mag.}, vol.~7, no.~2, pp. 71--84, Apr. 2015.

\bibitem{phan2018WSAadaptive}
D.-T. Phan-Huy, S.~Wesemann, J.~Bjoersell, and M.~Sternad, ``{Adaptive massive
  MIMO for fast moving connected vehicles: It will work with predictor
  antennas!}'' in \emph{Proc. IEEE WSA}, Bochum, Germany, Mar. 2018, pp. 1--8.

\bibitem{Apelfrojd2018PIMRCkalman}
R.~{Apelfr\"ojd}, J.~{Bj\"orsell}, M.~{Sternad}, and D.~{Phan-Huy}, ``Kalman
  smoothing for irregular pilot patterns; a case study for predictor antennas
  in {TDD} systems,'' in \emph{Proc. IEEE PIMRC}, Bologna, Italy, Sep. 2018,
  pp. 1--7.

\bibitem{madapatha2020integrated}
C.~Madapatha, B.~Makki, C.~Fang, O.~Teyeb, E.~Dahlman, M.-S. Alouini, and
  T.~Svensson, ``{On integrated access and backhaul networks: current status
  and potentials},'' Jun. 2020, [Online]. Available:
  \url{https://arxiv.org/abs/2006.14216}.

\bibitem{BJ2017PIMRCpredictor}
J.~Bj\"orsell, M.~Sternad, and M.~Grieger, ``Predictor antennas in action,'' in
  \emph{Proc. IEEE PIMRC}, Montreal, Quebec, Canada, Oct. 2017, pp. 1--7.

\bibitem{Guo2019WCLrate}
H.~{Guo}, B.~{Makki}, and T.~{Svensson}, ``Rate adaptation in predictor antenna
  systems,'' \emph{IEEE Wireless Commun. Lett.}, vol.~9, no.~4, pp. 448--451,
  Apr. 2020.

\bibitem{guo2020power}
H.~Guo, B.~Makki, M.-S. Alouini, and T.~Svensson, ``{Power allocation in
  HARQ-based predictor antenna systems},'' Apr. 2020, Accepted in IEEE Wireless
  Commun. Lett., [Online]. Available: https://arxiv.org/abs/2004.01421.

\bibitem{guo2020semilinear}
------, ``{A semi-linear approximation of the first-order Marcum $Q$-function
  with application to predictor antenna systems},'' Jan. 2020, [Online].
  Available: https://arxiv.org/abs/2001.09264.

\bibitem{Guo2020rate}
------, ``{On delay-limited average rate of HARQ-based predictor antenna
  systems},'' Apr. 2020, [Online]. Available: https://arxiv.org/abs/2004.01423.

\bibitem{polyanskiy2010channel}
Y.~Polyanskiy, H.~V. Poor, and S.~Verd{\'u}, ``Channel coding rate in the
  finite blocklength regime,'' \emph{IEEE Trans. Inf. Theory}, vol.~56, no.~5,
  pp. 2307--2359, May 2010.

\bibitem{bilstrup2008evaluation}
K.~Bilstrup, E.~Uhlemann, E.~G. Strom, and U.~Bilstrup, ``{Evaluation of the
  IEEE 802.11 p MAC method for vehicle-to-vehicle communication},'' in
  \emph{Proc. IEEE VTC}, Sep. 2008, pp. 1--5.

\bibitem{Makki2014WCLfinite}
B.~{Makki}, T.~{Svensson}, and M.~{Zorzi}, ``Finite block-length analysis of
  the incremental redundancy {HARQ},'' \emph{IEEE Wireless. Commun. Lett.},
  vol.~3, no.~5, pp. 529--532, Oct. 2014.

\bibitem{makki2013TCfeedback}
B.~Makki and T.~Eriksson, ``{Feedback subsampling in temporally-correlated
  slowly-fading channels using quantized CSI},'' \emph{IEEE trans. commun.},
  vol.~61, no.~6, pp. 2282--2294, Jun. 2013.

\bibitem{Shin2003TITcapacity}
H.~Shin and J.~H. Lee, ``Capacity of multiple-antenna fading channels: spatial
  fading correlation, double scattering, and keyhole,'' \emph{IEEE Trans. Inf.
  Theory}, vol.~49, no.~10, pp. 2636--2647, Oct. 2003.

\bibitem{Bocus2013CLapproximation}
M.~Z. {Bocus}, C.~P. {Dettmann}, and J.~P. {Coon}, ``{An approximation of the
  first order Marcum $Q$-function with application to network connectivity
  analysis},'' \emph{IEEE Commun. Lett.}, vol.~17, no.~3, pp. 499--502, Mar.
  2013.

\bibitem{makki2016TCwireless}
B.~Makki, T.~Svensson, and M.~Zorzi, ``{Wireless energy and information
  transmission using feedback: Infinite and finite block-length analysis},''
  \emph{IEEE Trans. Commun.}, vol.~64, no.~12, pp. 5304--5318, Dec. 2016.

\bibitem{corless1996lambertw}
R.~M. Corless, G.~H. Gonnet, D.~E. Hare, D.~J. Jeffrey, and D.~E. Knuth, ``{On
  the Lambert $\mathcal{W}$ function},'' \emph{{Advances in Computational
  Mathematics}}, vol.~5, no.~1, pp. 329--359, Dec. 1996.

\bibitem{makki2015TCfinite}
B.~Makki, T.~Svensson, and M.~Zorzi, ``{Finite block-length analysis of
  spectrum sharing networks using rate adaptation},'' \emph{IEEE Trans.
  Commun.}, vol.~63, no.~8, pp. 2823--2835, Aug. 2015.

\end{thebibliography}

\end{document}